\documentclass{article}
\usepackage{amssymb,amsmath}
\usepackage{graphicx,enumerate}

\title{The interest rate for saving as a possibilistic risk}

\author{Irina Georgescu \\ \footnotesize Bucharest University of Economics\\ \footnotesize Department of Economic Informatics and Cybernetics\\ \footnotesize Pia$\c{t}$a Romana No 6  R 70167, Oficiul Postal 22, Bucharest, Romania\\
 \footnotesize Email: irina.georgescu@csie.ase.ro \\ and \\ Jani Kinnunen \\  \footnotesize
$\AA$bo Akademi University, Tuomiokirkkotori 3,
Turku, 20500, Finland\\
\footnotesize Email: jani.kinnunen@abo.fi}

\date{}

\begin{document}
\maketitle

\begin{abstract}
In the paper there is studied an optimal saving model in which the interest-rate risk for saving is a fuzzy number. The total utility of consumption  is defined by using a concept of possibilistic expected utility. A notion of possibilistic precautionary saving is introduced as a measure of the variation of optimal saving level when moving from a sure saving model to a possibilistic risk model. A first result establishes a necessary and sufficient condition that the presence of a possibilistic interest-rate risk generates an extra-saving. This result can be seen as a possibilistic version of a Rothschilld and Stiglitz theorem on a probabilistic model of saving. A second result of the paper studies the variation of the optimal saving level when moving from a probabilistic model (the interest-rate risk is a random variable) to a possibilistic model (the interest-rate risk is a fuzzy number).

\end{abstract}


\newtheorem{definitie}{Definition}[section]
\newtheorem{propozitie}[definitie]{Proposition}
\newtheorem{remarca}[definitie]{Remark}
\newtheorem{exemplu}[definitie]{Example}
\newtheorem{intrebare}[definitie]{Open question}
\newtheorem{lema}[definitie]{Lemma}
\newtheorem{teorema}[definitie]{Theorem}
\newtheorem{corolar}[definitie]{Corollary}

\newenvironment{proof}{\noindent\textbf{Proof.}}{\hfill\rule{2mm}{2mm}\vspace*{5mm}}



% \fi

\section{Introduction}

The study of risk effects on saving decisions started with the papers of Leland, \cite{leland}, Sandmo \cite{sandmo} and Dr$\grave{e}$ze and Modigliani \cite{dreze}. They investigated the way the presence of a labor-income risk influences the increase or the decrease in the optimal saving level. At the same time, Sandmo \cite{sandmo} and Rothschild and Stiglitz \cite{rothschild} studied into the changes caused by an interest-rate risk on the optimal saving.

The notion of precautionary saving was introduced as a measure of the variation of optimal saving when moving from a certain optimal saving model to a saving model with risk. The positivity of precautionary saving indicates the fact that the consumer increases the optimal saving level in the presence of a risk. In case of a saving model with labor-income risk, the precautionary saving is positive if and only if the third derivative of the utility function is positive. By Kimball \cite{kimball}, \cite{kimball1} this last condition tells us that the consumer is prudent. If an interest-rate risk for saving is present in the model, then by a Rothschild and Stiglitz theorem from \cite{rothschild}, the precautionary saving is positive if and only if the relative prudence index is higher than $2$. An new proof of this theorem was given by Magnani in \cite{magnani2}.

Further on, the theme of the effects of the two types of risk (labor-income risk and interest-rate risk) on precautionary saving has been developed by Eeckhoudt and Schlesinger \cite{eeckhoudt1}, Li \cite{li}, Baiardi et al. \cite{baiardi}, etc. An intuitive exposure of the theme described above can be found in section $6$ of the monograph \cite{eeckhoudt}.

On the other hand, Zadeh's possibility theory \cite{zadeh} offers another way to model risk (by \cite{dubois}, \cite{carlsson2}, \cite{georgescu1}). Thereby risk can be thought of as a possibility distribution, particularly, a fuzzy number. Possibilistic approaches to themes of risk theory can be found in monographs \cite{carlsson2}, \cite{georgescu1}, as well as in a series of papers \cite{casademunt}, \cite{georgescu2}, \cite{georgescu3}, \cite{georgescu4} etc. In particular, the paper \cite{casademunt} develops an optimal saving model in which the labor-income risk is a fuzzy number.

As a continuation of paper \cite{casademunt}, in this paper we deal with an optimal saving model in which the interest-rate risk for saving is a fuzzy number. Following a parallel line with \cite{rothschild}, \cite{eeckhoudt} or \cite{eeckhoudt1}, the construction of our model is situated inside the possibilistic EU-theory from \cite{georgescu3}.

The total utility function of consumption is defined using the possibilistic expected utility in \cite{georgescu3}. We define a new notion of possibilistic precautionary saving and we characterize its positivity in terms of relative prudence index. This result can be regarded as a possibilistic variant of Rothschild and Stiglitz theorem from \cite{rothschild}, \cite{magnani1}.

We briefly present the content of the paper. In Section $2$ there are recalled from \cite{carlsson1}, \cite{carlsson2}, \cite{georgescu1}, \cite{georgescu2} the definitions of possibilistic expected utility, the main indicators associated with a fuzzy number (the possibilistic expected value, the possibilistic variance), as well as basic properties of these indicators. Section $3$ contains the construction of the optimal saving model in which the rate of interest for saving is a fuzzy number. In Section $4$ a notion of possibilistic precautionary saving is defined and a Rothschild and Stiglitz theorem regarding its positivity is proved.

\section{Possibilistic EU-theory}

Risk appears in several economic and financial problems. The study of different topics of risk theory was done in the framework of the classical von Neumann-Morgenstern theory (=EU-theory).  The central notion of $EU$-theory is the expected utility: if $X$ is a random variable and $u: {\bf{R}} \rightarrow {\bf{R}}$ is a continuous function, then the mean value $M(u(X))$ \footnote{If $X$ is a random variable (with respect to the probability space $(\Omega, \mathcal{K}, P)$) then we denote by $M(X)$ its mean value and by $D^2(X)$ its variance.} of the random variable $u(X)$ is the expected utility associated with $X$ and $u$. For some particular forms of $u$ the main indicators associated with $X$ are obtained: the expected value, the variance, the moments, etc. Accordingly,  von Neumann-Morgenstern EU-theory has two fundamental entities: the random variable $X$ representing the risk and the utility function $u$ representing the decision-maker (the agent).

The possibilistic $EU$-theory from \cite{georgescu1}, \cite{georgescu3} starts from three basic entities:

$\bullet$ a weighting function $f: [0, 1] \rightarrow {\bf{R}}$ ($f$ is a non-negative and increasing function that satisfies the condition $\int_0^1 f(\gamma) d\gamma=1$).

$\bullet$ a fuzzy number $A$ whose level sets have the form $[A]^\gamma=[a_1(\gamma), a_2(\gamma)]$ for any $\gamma \in [0, 1]$

$\bullet$ a utility function $u: {\bf{R}} \rightarrow {\bf{R}}$ of class $\mathcal{C}^2$.

In case of a possibilistic $EU$-theory risk is represented by a fuzzy number $A$ and the utility function describes the agent.

We fix a weighting function $f$, a fuzzy number $A$ and a utility function $u$.

\begin{definitie}
\cite{georgescu1}, \cite{georgescu3} The possibilistic expected utility $E_f(u(A))$ associated with the triple $(f, A, u)$ is defined by:

$E_f(u(A))=\frac{1}{2}\int_0^1 [u(a_1(\gamma))+u(a_2(\gamma))]f(\gamma)d\gamma$ (2.1)

Setting in (2.1) $u=1_{\bf{R}}$ (the identity function of $\bf{R}$) we obtain the possibilistic expected value (\cite{carlsson1}, \cite{carlsson2}, \cite{georgescu1}):

$E_f(A)=\frac{1}{2}\int_0^1 [a_1(\gamma)+a_2(\gamma)]f(\gamma) d\gamma$  (2.2)

If we take in (2.1) $u(x)=(x-E_f(A))^2$ then one obtains the possibilistic variance \cite{carlsson1}, \cite{carlsson2}, \cite{georgescu1}:

$Var_f(A)=\frac{1}{2}\int_0^1 [(a_1(\gamma)-E_f(A))^2+(a_2(\gamma)-E_f(A))^2]f(\gamma)d\gamma$ (2.3)
\end{definitie}

\begin{propozitie}
\cite{georgescu1}, \cite{georgescu3} Let $a, b \in \bf{R}$ and $g: \bf{R} \rightarrow \bf{R}$, $h: \bf{R} \rightarrow \bf{R}$ two utility functions. We denote
$u=ag+bh$. Then for any fuzzy number $A$ we have

$E_f(u(A))=aE_f(g(A))+bE_f(h(A))$  (2.4)
\end{propozitie}

\begin{propozitie}
(\cite{georgescu1} the proof of Proposition 4.4.2, \cite{georgescu3}) For any fuzzy number $A$ and any utility function $u$ the following approximation formula holds:

$E_f(u(A)) \approx u(E_f(A))+\frac{1}{2}u''(E_f(A))Var_f(A)$  (2.5)
\end{propozitie}

\section{Saving in the presence of an interest-rate risk}

The construction of the possibilistic model for interest rate on saving from this section will have as a starting point the two optimal saving models from
\cite{eeckhoudt}, Section 6.2, the certain model and the probabilistic model. To fix ideas we will briefly present these two models:

{\emph{The certain model}} This is a two-period consumption model build on the following entities:

$\bullet$ the consumer has the same utility function $u(x)$ both in the first period (denoted $0$) and in the second period (denoted $1$);

$\bullet$ $y_0$ is the sure level of first-period income;

$\bullet$ $r$ is the interest rate for saving;

$\bullet$ $R=1+r>0$ is the return on saving;

$\bullet$ $s>0$ is the level of saving.

As usual, one assumes that the utility function $u$ is of class $\mathcal{C}^2$ and $u'>0$, $u''<0$.

The total utility function of consumption will be:

$U(s)=u(y_0-s)+u(s(1+r))=u(y_0-s)+u(sR)$  (3.1)

The function $U(s)$ is strictly concave. The consumer wants to determine that value $s^\ast$ of $s$ maximizing the total utility $U(s)$:

$\displaystyle \max_s U(s)=U(s^\ast)$  (3.2)

The optimal saving $s^\ast$ is a solution of the first-order condition

$-u'(y_0-s)+Ru'(sR)=0$ (3.3)

Since $U'(s)$ is strictly decreasing, the solution $s^\ast$ of (3.3) is unique (when it exists.)

{\emph{The probabilistic model}} In this model, instead of an interest rate for saving $r$ we have a random variable $\tilde r$ with $M(\tilde r)=r$. \footnote{The other entities of the probabilistic model are the same as in the certain model.} The return of the saving will be the random variable $\tilde R=1+ \tilde r >0$. One notices that $M(\tilde R)=R$. The total utility function of the model will be:

$V(s)=u(y_0-s)+M(u(s \tilde R))$ (3.4)

$V(s)$ is also strictly concave. We denote by $s^\ast_1$ the solution of the maximization problem $\displaystyle \max_s U(s)$.

The first order condition associated with the maximization problem will be:

$-u'(y_0-s)+M(\tilde R u'(s\tilde R))=0$ (3.5)

The (probabilistic) precautionary saving $s^\ast_1-s^\ast$ describes the variation of optimal saving when moving from the certain model to the probabilistic model.

\begin{propozitie}
(Rothschild and Stiglitz, \cite{rothschild}) If $u$ is of class $\mathcal{C}^2$ and $u'>0$, $u''<0$, then the following equivalences hold:

(i) $s^\ast_1-s^\ast \geq 0$;

(ii) $-\frac{u'''(Rs^\ast)}{u''(R s^\ast)} R s^\ast \geq 2$.
\end{propozitie}

We consider Kimball's absolute prudence index $P_u(x)=-\frac{u'''(x)}{u''(x)}$, $x \in {\bf{R}}$ (by \cite{kimball}, \cite{kimball1}).

The function $RP_u(x)=-x\frac{u'''(x)}{u''(x)}$, $x \in {\bf{R}}$ is called the relative prudence index. Then the left hand side of the inequality (ii) is exactly $RP_u(R s^\ast)$.

{\emph{The possibilistic  model}} We fix a weighting function $f: [0, 1] \rightarrow \bf{R}$. In case of this model, the uncertainty from period $0$ is no longer probabilistic, but possibilistic: interest rate for saving will be a fuzzy number $B$ such that $E_f(B)=r$. The return on saving will be the fuzzy number $A=1+B$, thus $E_f(A)=1+r=R$. We will make the hypothesis that $supp(A)=\{ x \in {\bf{R}} |A(x)>0 \}$ is a subset of the interval $(0, \infty)$. We denote by $[A]^\gamma=[a_1(\gamma), a_2(\gamma)]$, $\gamma \in [0, 1]$ the level sets of the fuzzy number $A$.

In rest, the entities of the possibilistic model coincide with the ones of the two models presented above. We introduce a new function $v: {\bf{R}}^2 \rightarrow {\bf{R}}$:

$v(s, x)=u(sx)$  (3.6)

The total utility function of the possibilistic model will be:

$W(s)=u(y_0-s_0)+E_f(v(s, A))$ (3.7)

$E_f(v(s, A))$ is the possibilistic expected utility associated with $f$, the fuzzy number $A$ and the utility function $v(s, .)$ in which $s$ is a parameter. By (2.1), the explicit form of $W(s)$ will be:

$W(s)=u(y_0-s)+\frac{1}{2}\int_0^1 [v(s, a_1(\gamma))+v(s, a_2(\gamma))]f(\gamma)d\gamma$

$=u(y_0-s)+\frac{1}{2}\int_0^1 [u(sa_1(\gamma))+u(sa_2(\gamma))]f(\gamma)d\gamma$

From (3.6) one obtains

$\frac{\partial v(s,x)}{\partial s}=x u'(sx)$ (3.8)

$\frac{\partial^2 v(s,x)}{\partial s^2}=x^2 u''(sx)<0$ (3.9)

Deriving the explicit form of $W(s)$ one obtains:

$W'(s)=-u'(y_0-s)+\frac{1}{2}\int_0^1 [a_1(\gamma)u'(sa_1(\gamma))+a_2(\gamma)u'(sa_2(\gamma))]f(\gamma)d\gamma$

By (3.8), $W'(s)$ is written as follows:

$W'(s)=-u'(y_0-s)+E_f[\frac{\partial v(s, A)}{\partial s}]$ (3.10)

An analogous calculation leads to

$W''(s)=u''(y_0-s)+E_f[\frac{\partial^2 v(s, A)}{\partial s^2}]$ (3.11)

In formula (3.10), $E_f[\frac{\partial v(s, A)}{\partial s}]$ is the possibilistic expected utility associated with $f$, $A$ and the utility function $\frac{\partial v(s, .)}{\partial s}$, and in (3.11) $E_f[\frac{\partial^2 v(s, A)}{\partial s^2}]$ is the possibilistic expected utility associated with
$f$, $A$ and the utility function $\frac{\partial^2 v(s, .)}{\partial s^2}$.

\begin{propozitie}
The total utility function $W$ is strictly concave.
\end{propozitie}

\begin{proof}
By hypothesis, $u''<0$  and (3.9), (3.11) it follows $W''(s) <0$ for any $s$.

In case of the possibilistic model, the consumer wishes to obtain the solution $s^{\ast \ast}$ of the optimization problem:

$\displaystyle \max_s W(s)$  (3.12)

The optimal saving $s^{\ast \ast}$ will be the solution of the first-order condition $W'(s)=0$. By (3.10), the first-order condition $W'(s)=0$ is written:

$-u'(y_0-s)+E_f[\frac{\partial v(s, A)}{\partial s}]=0$ (3.13)
\end{proof}

\section{The possibilistic precautionary saving}

We recall from the previous section that $s^\ast$ is the optimal saving for the certain model, $s^\ast_1$ is the optimal saving for the probabilistic model, and $s^{\ast \ast}$ is the optimal saving for the possibilistic model (3.12). According to the first-order conditions (3.3) and (3.14) the following equalities hold:

$u'(y_0-s^\ast)=Ru'(s^\ast R)$ (4.1)

$u'(y_0-s^{\ast \ast})=E_f[\frac{\partial v(s^{\ast \ast}, A)}{\partial s}]$ (4.2)

The difference $s^{\ast \ast}-s^\ast$, called possibilistic precautionary saving, represents the variation of optimal saving when moving from the certain model (3.1) to the possibilistic model (3.13).

The following proposition is the possibilistic version of the Rothschild and Stiglitz (see Proposition 3.1).

\begin{propozitie}
Assume that the utility function $u$ is of class $\mathcal{C}^3$ and $u'>0, u''<0$. Then the following assertions are equivalent:

(a) $s^{\ast \ast}-s^\ast \geq 0$;

(b) $-\frac{u'''(R s^\ast)}{u''(R s^\ast)} Rs^\ast \geq 2$.
\end{propozitie}

\begin{proof}
By Proposition 3.2, the derivative $W'$ of the total utility $W$ is a strictly decreasing function. Then, taking into account that $W'(s^{\ast \ast})=0$, the following equivalence holds:

$s^{\ast \ast} \geq s^\ast $ iff $0=W'(s^{\ast \ast}) \leq W'(s^\ast)$  (4.3)

From (3.10) one obtains

$W'(s^\ast)=-u'(y_0-s^\ast)+E_f[\frac{\partial v(s^\ast, A)}{\partial s}]$ (4.4)

We consider the following function:

$h(x)=\frac{\partial v(s^\ast, x)}{\partial s}=xu'(s^\ast x)$ (4.5)

Then

$h'(x)=u'(s^\ast x)+s^\ast x u''(s^\ast x)$  (4.6)

$h''(x)=s^\ast [2u''(s^\ast x)+x s^\ast u'''(s^\ast x)]$ (4.7)

We apply Proposition 2.3 to the function $h$ and the fuzzy number $A$:

$E_f(h(A)) \approx h(E_f(A))+\frac{1}{2}h''(E_f(A))Var_f(A)$.

Since $E_f(A)=R$ and $E_f(h(A))=E_f[\frac{\partial v(s^\ast, A)}{\partial s}]$, the previous approximation formula becomes:

$E_f[\frac{\partial v(s^\ast, A)}{\partial s}] \approx h(R)+\frac{1}{2}h''(R)Var_f(A)$

Taking into account (4.5) and (4.7)

$E_f[\frac{\partial v(s^\ast, A)}{\partial s}] \approx R u'(s^\ast R)+\frac{Var_f(A)}{2}s^\ast [2u''(s^\ast R)+R s^\ast u'''(s^\ast R)]$  (4.8)

Replacing in (4.4) it follows

$W'(s^\ast) \approx -u'(y_0-s^\ast)+Ru'(s^\ast R)+\frac{Var_f(A)}{2}s^\ast [2u''(s^\ast R)+R s^\ast u'''(s^\ast R)]$

Taking into account (4.1), the previous formula becomes:

$W'(s^\ast) \approx \frac{Var_f(A)}{2}s^\ast [2u''(s^\ast R)+R s^\ast u'''(s^\ast R)]$  (4.9)

But $Var_f(A)>0$ and $s^\ast>0$ by hypothesis, thus

$0 \leq W'(s^\ast)$ iff $2u''(s^\ast R)+s^\ast R u'''(s^\ast R)\geq 0$ (4.10)

Taking into account that $u''(s^\ast R)<0$, from the equivalences (4.3) and (4.10) it follows immediately

$s^{\ast \ast} \geq s^\ast$ iff $-\frac{u'''(R s^\ast)}{u''(R s^\ast)}Rs^\ast \geq 2$.

\end{proof}

\begin{remarca}
The assertions (a), (b) from Proposition 4.1 remain equivalent when the inequalities are replaced by strict inequalities: $s^{\ast \ast} - s^\ast>0$ iff $-\frac{u'''(R s^\ast)}{u''(R s^\ast)}Rs^\ast > 2$.
\end{remarca}

\begin{remarca}
The condition $s^{\ast \ast} - s^\ast>0$ means that the presence of possibilistic interest-rate risk in an optimal saving model draws by itself the increase in the optimal saving level (we will say that the possibilistic interest-rate generates an extra-saving).
\end{remarca}

One notices that condition (ii) from Proposition 3.1 coincides with condition (b) from Proposition 4.1. Combining the two propositions, the following result is obtained:

\begin{corolar}
Assume that $u$ is of class $\mathcal{C}^3$ and $u'>0, u''<0$. Then the following assertions are equivalent:

(a) $s^{\ast \ast}-s^\ast \geq 0$ (resp. $s^{\ast \ast}-s^\ast > 0$);

(b) $s^\ast_1-s^\ast \geq 0$ (resp. $s^\ast_1-s^\ast > 0$);

(c) $-\frac{u'''(R s^\ast)}{u''(R s^\ast)} Rs^\ast \geq 2$ (resp. $-\frac{u'''(R s^\ast)}{u''(R s^\ast)} Rs^\ast >2$).

\end{corolar}

By the previous corollary, the probabilistic interest-rate risk generates an extra-saving if and only if the possibilistic interest-rate risk generates an extra-saving.

One asks how to compare the optimal saving $s^\ast_1$ (corresponding to the probabilistic model) with the optimal saving $s^{\ast \ast}$ (corresponding to the possibilistic model).

\begin{propozitie}
Assume that $u$ is of class $\mathcal{C}^3$ and $u'>0, u''<0$. Then the following assertions are equivalent:

(i) $s^{\ast \ast} \geq s^\ast_1$;

(ii) $[2u''(s^\ast_1 R)+Rs^\ast_1u'''(s^\ast_1 R)][Var_f(A)-D^2(\tilde R)] \geq 0$.
\end{propozitie}

\begin{proof}
From (3.10) and (3.5) it follows that the optimal saving $s^\ast_1$ satisfies the following equalities:

$W'(s^\ast_1)=-u'(y_0-s^\ast_1)+E_f[\frac{\partial v(s^\ast_1, A)}{\partial s}]$ (4.11)

$u'(y_0-s^\ast_1)=M[\tilde R u'(s^\ast_1 \tilde R)]$  (4.12)

By Proposition 3.2, $W'$ will be a strictly decreasing function. Then

$s^{\ast \ast} \geq s^\ast_1$ iff $0 \leq W'(s^{\ast \ast}) \leq W'(s^\ast_1)$ (4.13)

Taking into account (4.12), equality (4.11) becomes:

$W'(s^\ast_1)=-M[\tilde R u'(s^\ast_1 \tilde R)]+E_f[\frac{\partial v(s^\ast_1, A)}{\partial s}]$ (4.14)

Taking into account (4.8) we will have the approximation formula:

$E_f[\frac{\partial v(s^\ast_1, A)}{\partial s}] \approx Ru'(s^\ast_1 R)+\frac{Var_f(A)}{2}s^\ast_1 [2u''(s^\ast_1 R)+R s^\ast_1 u'''(s^\ast_1 R)]$ (4.15)

For the computation of the other term of the sum from (4.14) we consider the following function:

$g(x)=x u'(s^\ast_1 x)$ (4.16)

A simple computation shows that

$g''(x)=s^\ast_1 [2 u''(s^\ast_1 x)+x s^\ast_1 u'''(s^\ast_1 x)]$ (4.17)

Applying to $g$ the second-order Taylor approximation, the known approximation formula is obtained:

$M[\tilde R u'(s^\ast_1 \tilde R)]=M[g(\tilde R)] \approx g(M(\tilde R))+\frac{1}{2}g''(M(\tilde R)) D^2(\tilde R)$ (4.18)

Since $M(\tilde R)=R$ from (4.18) one obtains:

$M[\tilde R u'(s^\ast_1 \tilde R)] \approx g(R)+\frac{g''(R)}{2}D^2(\tilde R)$ (4.19)

From (4.16), (4.17) and (4.19) we find the approximation

$M[\tilde R u'(s^\ast_1 \tilde R)] \approx R u'(s^\ast_1 R)+\frac{D^2(\tilde R)}{2}s^\ast_1 [2u''(s^\ast_1 R)+R s^\ast_1 u'''(s^\ast_1 R)]$ (4.20)

Using (4.14), (4.15) and (4.20) it follows

$W'(s^\ast_1) \approx \frac{s^\ast_1}{2}[2 u''(s^\ast_1 R)+Rs^\ast_1 u'''(s^\ast_1 R)] [Var_f(A)-D^2(\tilde R)]$ (4.21)

Since $s^\ast_1>0$, from (4.13) and (4.20) the equivalence of conditions (i) and (ii) follows immediately.

\end{proof}

\begin{remarca}
The assertions (i), (ii) from Proposition 4.5 remain equivalent if we replace the inequality by strict inequality. Condition $s^{\ast \ast}>s^\ast_1$ expresses the fact that moving from the probabilistic optimal saving model to the possibilistic optimal saving model, an extra-saving is obtained.
\end{remarca}

\begin{corolar}
Under the conditions of Proposition 4.5, the following are equivalent:

(i) $s^{\ast \ast} > s^\ast_1$;

(ii) The disjunction of the following assertions holds:

$(ii_1)$ $-\frac{u'''(Rs^\ast_1)}{u''(Rs^\ast_1)} R s^\ast_1 >2$ and $Var_f(A)>D^2(\tilde R)$

$(ii_2)$ $-\frac{u'''(Rs^\ast_1)}{u''(Rs^\ast_1)} R s^\ast_1 <2$ and $Var_f(A)<D^2(\tilde R)$
\end{corolar}

\begin{proof}
From Proposition 4.5 and Remark 4.6, it follows that $s^{\ast \ast} > s^\ast_1$ if and only if the disjunction of the following assertions is true:

(a) $2u''(s^\ast_1R)+R s^\ast_1 u'''(s^\ast_1 R)>0$ and $Var_f(A)-D^2(\tilde R)>0$.

(b) $2u''(s^\ast_1R)+R s^\ast_1 u'''(s^\ast_1 R)<0$ and $Var_f(A)-D^2(\tilde R)<0$.

Taking into account that $u''(s^\ast_1 R)<0$ it follows $2u''(s^\ast_1 R)+R s^\ast_1 u'''(s^\ast_1 R)>0$ iff $-\frac{u'''(Rs^\ast_1)}{u''(Rs^\ast_1)} R s^\ast_1 >2$ thus the assertions (a) and $(ii_1)$ are equivalent. Similarly, the assertions (b) and $(ii_2)$ are equivalent.

\end{proof}

\begin{remarca}
By Corollary 4.7, the fact that moving from the optimal saving model with probabilistic risk to the optimal saving model with possibilistic risk generated an extra-saving is characterized by inequalities expressed in terms of (probabilistic and possibilistic) relative prudence index $-\frac{u'''(Rs^\ast_1)}{u''(Rs^\ast_1)} R s^\ast_1$ and the variances $E_f(A)$ and $D^2(\tilde R)$.
\end{remarca}

\begin{exemplu}
\normalfont
Assume that the consumer's utility function is of type $CRRA$: $u(x)=\frac{x^{1-\gamma}}{1-\gamma}$ for any $x>0$ ($\gamma$ is a parameter with the property $\gamma >0$ and $\gamma \not =1$). Then the relative prudence index is $R P_u(x)=\gamma+1$, for any $x$ (see \cite{eeckhoudt}, p.99).

Applying Remark 4.2 to this case, the possibilistic interest-rate risk $A$ generates an extra-saving if and only if $\gamma >1$.

We fix the weighting function $f(t)=2t$, for $t \in [0, 1]$.

Assume that $0< c <d$. We consider the fuzzy number $A$ with $a_1(\gamma)=c$, $a_2(\gamma)=d$ for any $t \in [0, 1]$ and $\tilde R$ the uniformly distributed random variable on the interval $[c, d]$. A simple computation shows that

$R=E_f(A)=M(\tilde R)=\frac{c+d}{2}$;

$Var_f(A)=\frac{(c-d)^2}{4}$; $D^2(\tilde R)=\frac{(c-d)^2}{12}$.

Noticing that $Var_f(A)<D^2(\tilde R)$, by applying Corollary 4.7 it follows that $s^{\ast \ast} >s^\ast_1$ if and only if $\gamma <1$.

\end{exemplu}

\section{Concluding Remarks}
Using the notion of possibilistic expected utility from \cite{georgescu3}, in this paper an optimal saving model is built in which an interest rate risk modeled by a fuzzy number is present.

Two main results are proved:

$\bullet$ a necessary and sufficient condition for the presence of a possibilistic interest rate risk to generate an extra-saving;

$\bullet$ the comparison of optimal saving levels corresponding to two models with different descriptions of the probabilistic interest-rate risk model from \cite{magnani2} and the possibilistic model from this paper.

The first result can be seen as a possibilistic version of a Rothschild and Stiglitz theorem from \cite{magnani2}.

We mention a few open problems:

(1) The study of possibilistic models of optimal saving in which to appear both labor income risk and interest rate risk, one of them in probabilistic form (random variable), the other in possibilistic form (fuzzy number);

(2) Continuing the research from the paper \cite{georgescu5} by studying some possibilistic and mixed models of optimal saving inspired from the models from papers \cite{gunning}, \cite{vergara}. These models will be obtained as combinations of four types of possibilistic risk:  labor income risk, wealth risk, asset risk and capital risk;

(3) Building some optimal saving models in a larger possibilistic frame, such as the one provided by expected utility operators (see \cite{georgescu1}, Chapter $5$).

\end{document}